\documentclass[letterpaper, 10 pt, conference]{ieeeconf}  

\IEEEoverridecommandlockouts                              

\usepackage{amssymb}
\usepackage{amsthm}

\usepackage{graphics} 
\usepackage{float}
\usepackage{epsfig} 
\usepackage{times} 
\usepackage{amsmath, bm} 
\usepackage{amsfonts}  
\usepackage{comment}
\usepackage{xcolor}
\usepackage{cite}
\usepackage[english]{babel}
\usepackage{subfigure}
\usepackage[makeroom]{cancel}
\usepackage{comment}
\usepackage{caption}

\newcommand{\cX}{{\cal X}}

\newcommand{\cC}{{\cal C}}

\newcommand{\cS}{{\cal S}}
\newcommand{\cD}{{\cal D}}

\newcommand{\mR}{{\mathbb R}}

\newcommand{\mU}{{\mathbb U}}

\newcommand{\bs}{{\mathbf s}}

\newcommand{\by}{{\mathbf y}}
\newcommand{\bff}{{\mathbf f}}

\newcommand{\bw}{{\mathbf w}}
\newcommand{\bx}{{\mathbf x}}

\newcommand{\bA}{{\mathbf A}}

\theoremstyle{plain}
\newtheorem{theorem}{\textbf{Theorem}}
\newtheorem{definition}{\textbf{Definition}}

\newtheorem{proposition}{\textbf{Proposition}}

\theoremstyle{definition}
\newtheorem{assumption}{Assumption}
\newtheorem{remark}{\textbf{Remark}}

\DeclareMathOperator*{\argmin}{arg\,min}

\begin{document}

\title{Extensions of the Path-integral formula for computation of Koopman eigenfunctions}

\author{Shankar A. Deka and Umesh Vaidya
\thanks{Shankar A. Deka is with the Department of Electrical Engineering and Automation, School of Electrical Engineering, Aalto University, 02150 Espoo, Finland (Email: {\tt\small shankar.deka@aalto.fi}). Umesh Vaidya are with the Dept. of Mechanical Engineering., Clemson University, Clemson SC. (Email: {\tt\small uvaidya@clemson.edu}. UV will like to acknowledge financial support from the NSF 2031573 grant.)
}
}

\maketitle

\begin{abstract}
Representing nonlinear dynamical systems using the Koopman Operator and its spectrum has distinct advantages in terms of linear interpretability of the model as well as in analysis and control synthesis through the use of well-studied techniques from linear systems theory. As such, efficient computation of Koopman eigenfunctions is of paramount importance towards enabling such Koopman-based constructions. To this end, several approaches have been proposed in literature, including data-driven, convex optimization, and Deep Learning-based methods. In our recent work, we proposed a novel approach based on path-integrals that allowed eigenfunction computations using a closed-form formula. In this paper, we present several important developments such as finite-time computations, relaxation of assumptions on the distribution of the principal Koopman eigenvalues, as well as extension towards saddle point systems, which greatly enhance the practical applicability of our method. 
\end{abstract}


\section{Introduction}\label{sec:Intro}
The use of operator theoretic methods involving the Koopman operator has become popular in the last decade. The Koopman operator provides for a linear (albeit possibly infinite dimensional) lifting of nonlinear system dynamics. Linear representation of nonlinear system dynamics leads to the development of systematic methods for data-driven identification, analysis, and control of nonlinear dynamics\cite{peitz2019koopman,villanueva2021towards,borggaard2009control, abraham2019active,korda2018linear, sootla2018optimal,fackeldey2020approximative,kaiser2021data,ma2019optimal,huang2022convex,deka2022koopman}. The main power of Koopman theory lies in its spectrum, i.e., eigenvalues and eigenfunctions. While the Koopman operator is used to study statistical or ensemble behavior of a dynamical system, the spectrum of the Koopman operator, particularly the principal spectrum, can be used to recover the geometrical structure of the underlying system. In particular, stable and unstable manifolds can be recovered as joint zero-level curves of Koopman eigenfunctions \cite{mezic2020spectrum}. These manifolds find their utility in reduced-order modeling, identifying stability boundaries, and optimizing the control design of dynamic systems \cite{vaidya2022spectral,umathe2023spectral,umathe2022reachability}. 

\begin{figure}[!htp]
    \centering
    \includegraphics[width=0.9\linewidth]{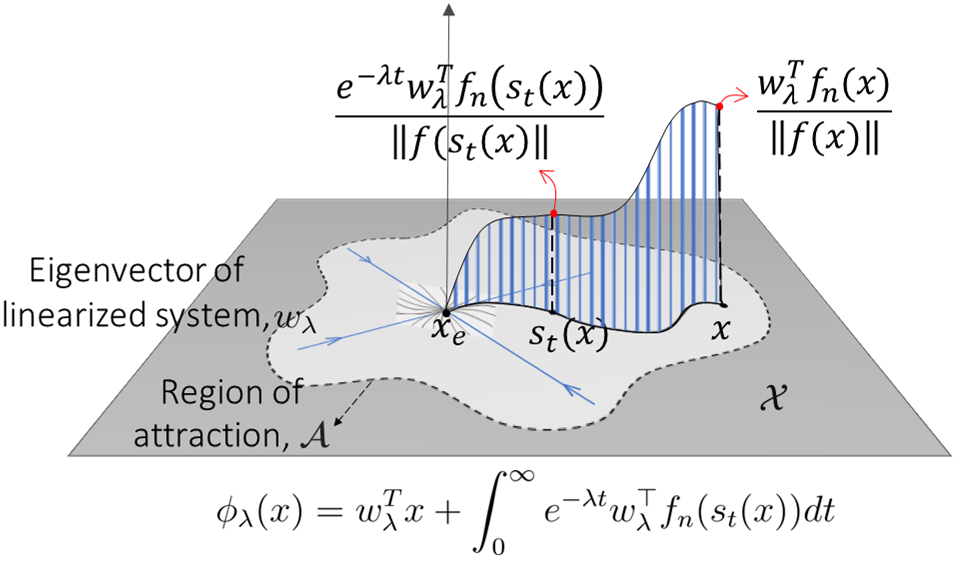}
    \caption{Path-integral formula for computing eigenfunction $\phi_\lambda$ corresponding to eigenvalue $\lambda$ \cite{deka2023path}. The integral term is visualized as the blue-shaded area. Please see section \ref{sec:PIF} for notations.}
    \label{fig:PIF}
    \vspace{-0.5cm}
\end{figure}

Given the significance of Koopman eigenfunctions, various computational methods have been developed to approximate it \cite{housparse,klus2020eigendecompositions,korda2018convergence,korda2020optimal,williams2015data}. Most existing methods rely on estimating a finite-dimensional representation of the Koopman operator for the approximate computation of the Koopman spectrum. While this approach is more streamlined, it may be computationally inefficient. Furthermore, the eigenfunctions approximated using the Koopman operator may lead to an inaccurate approximation of the principal eigenfunction of the Koopman operator. The principal eigenfunctions of the Koopman operator are associated with the eigenvalues of the linearization of the dynamical system at the origin. A system with a $n$ dimensional state space has only $n$ principal eigenfunctions. Thus, generating the Koopman spectrum directly from these $n$ principal eigenfunctions greatly reduces the computational burden. 

For all these reasons, it becomes essential to develop direct methods for approximating the principal spectrum of the Koopman operator. In \cite{deka2023path}, we proposed a path-integral formula for computation of the principal Koopman spectrum. We showed that the principal eigenfunctions of the Koopman operator satisfy a linear partial differential equation (PDE). A path-integral formula is then proposed for solving the PDE. One of the distinguishing features of the path-integral formula is that it does not involve parameterization via a choice of basis functions, and leads to the computation of eigenfunction value at any given point in the state space (figure \ref{fig:PIF}). The results in \cite{deka2023path} for eigenfunction computation were, however, only restricted to a system with a stable (or anti-stable) equilibrium point. \textcolor{black}{The path-integral formula for a stable equilibrium point is asymptotic in nature and involves computing the integral over infinite time. For a system with a saddle-type equilibrium point, the path integral computed over infinite time is not well defined.} 

In this paper, we extend the results from \cite{deka2023path} for systems with a saddle-type equilibrium point. \textcolor{black}{Computing the principal eigenfunctions of such systems has several applications. In \cite{umathe2023spectral}, the principal eigenfunctions of the Koopman operator associated with the saddle-type equilibrium point were used to characterize the stability boundary. In particular, the joint zero-level curve of eigenfunction associated with positive eigenvalue of the saddle equilibrium point forms the stable manifold and the stability boundary. As another example, Lagrangian submanifolds, which are fundamental objects used in the construction of solution to the Hamilton Jacobi equation, can be obtained using the principal eigenfunctions of the Koopman operator associated with corresponding Hamiltonian dynamical system \cite{vaidya2022spectral} (note that these Hamiltonian systems have a saddle point equilibrium).}

The main contribution of this paper is to provide conditions under which a finite-time path-integral helps approximate the Koopman principal eigenfunctions for a system with a saddle equilibrium point. We provide two methods for approximating Koopman eigenfunctions for a system with saddle-type equilibrium point. Our first method relies on prior knowledge of Koopman eigenfunctions in a certain subset of the state space. We propose using local Extended Dynamic Mode Decomposition (EDMD) for this. It is important to emphasize that the EDMD algorithm \cite{williams2015data} is only used locally in a small region of the state space to locally approximate the eigenfunctions. The computation of eigenfunction values globally in the rest of the state space still relies on the path-integral formula. In our second method, we propose a transformation of the original vector field far from the equilibrium point. The transformed vector field does not affect the value of the eigenfunction in the region containing the equilibrium point. Both these methods enable finite-time computation of the path-integral formula. Finally, we present simulation results to validate these two methods for eigenfunction computation. Our simulation results show that both these methods can approximate the eigenfunction accurately. We also demonstrate the application of these methods for the computation of optimal control following results developed in \cite{vaidya2022spectral},\cite{guo2022tutorial}.

\section{Preliminaries}
\noindent\textbf{Notations:} ${\rm Re}(\cdot)$ denotes the real part of its argument, which can be \textcolor{black}{a constant or a function}. The space of continuously differentiable functions on a set $\cX$ is denoted by $C^1(\cX)$. $\mathbb{R}$ denotes set of reals and $\mathbb{C}$ is the complex set. The notation $\|\cdot\|$ is used for 2-norm of a complex vector.\\

We consider autonomous systems in this paper, given by the dynamics
\begin{align}\label{eq:dynamics}
    \dot \bx=\bff(\bx),\;\;\bx\in \cX\subseteq \mR^n,
\end{align}
where $\bff : \cX \rightarrow \mR^n$ is a smooth function and $\cX$ is \textcolor{black}{compact} set containing the origin. We assume that the origin is a hyperbolic equilibrium point. Let $s_t:\cX \rightarrow \cX$ denote the flow-map of this system, that is,
$
    s_t(\bx) = \bx + \int_0^t \bff(\bx(s))ds,
$
for all $t\ge 0$ and $\bx\in \cX$.

\begin{definition}\big[\textup{Koopman Operator}\big]
Given the dynamical system \eqref{eq:dynamics}, the Koopman semigroup of operators is defined as the linear operators $\mathbb{U}_t$ acting on a (Banach) space $\mathcal{F}$ of functions $\psi: \cX \rightarrow \mathbb{C}$ such that
\begin{equation}\label{koopman_operator}
   \left[\mathbb{U}_t \psi\right](\bx) = \psi(s_t(\bx)),
\end{equation}
for every $\psi \in \mathcal{F}$.
\end{definition} 

We shall take the space $\mathcal{F}$ to be the Banach space $\cC^1(\cX)$ in this paper.
The Koopman semigroup $\mathbb{U}_t$ are linear operators, and one can correspondingly define their spectrum as follows.
\begin{definition}\big[\textup{Koopman eigenfunctions and eigenvalues}\big] A function $\phi_\lambda(\bx)\in \cC^1(\cX)$  is said to be an eigenfunction of the Koopman operator associated with eigenvalue $\lambda$ if
\begin{eqnarray}
[\mU_t \phi_\lambda](\bx)=e^{\lambda t}\phi_\lambda(\bx)\label{eig_koopman}
\end{eqnarray}
or equivalently,
\begin{align}
    \frac{\partial \phi_\lambda}{\partial \bx}{\bff}(\bx)=\lambda \phi_\lambda(\bx)\label{eig_koopmang},
\end{align}
\textcolor{black}{which is the PDE associated with the infinitesimal Koopman generator.}
\end{definition}

\section{Path-integral formula}\label{sec:PIF}
The spectrum of the Koopman operator, in general, is quite intricate and could consist of discrete and continuous parts. Furthermore, the spectrum also depends on the underlying functional space used in the approximation \cite{mezic2020spectrum}.
In this paper, we are interested in approximating the eigenfunctions of the Koopman operator with associated eigenvalues that are the same as that of the linearization of the system dynamics at the equilibrium point. With the hyperbolicity assumption on the equilibrium point of the system \eqref{eq:dynamics}, this part of the Koopman spectrum is known to be discrete and well-defined \cite{mezic2020spectrum}. 

Following this hyperbolicity assumption, one can write the system dynamics \eqref{eq:dynamics} as 
\begin{align}
\dot \bx=\bff(\bx)=\bA \bx+\bff_n(\bx),\label{eq:dynamics}
\end{align}
where $\bA :=\frac{\partial \bff}{\partial \bx}(0)$ is the Jacobian matrix at the origin  and $\bff_n(\bx):=\bff(\bx)-\bA \bx$ is the purely nonlinear part of the vector field $\bff(\bx)$. Let $\lambda$ be an eigenvalue of the linearization, i.e., $\bA$, and let $\phi_\lambda(\bx)$ be the eigenfunction associated with the eigenvalue $\lambda$ (such eigenfunctions are called principal eigenfunctions).
Similar to the system decomposition into linear and nonlinear parts, the principal eigenfunction, $\phi_\lambda(\bx)$, also admits a decomposition into linear and nonlinear terms as follows:
\begin{align}
\phi_\lambda(\bx)=\bw_\lambda^\top \bx+h_\lambda(\bx), \label{princ_eig}
\end{align}
where $\bw_\lambda^\top\bx$ is the linear part and $h_\lambda(\bx)$ is the purely nonlinear term and hence satisfies $\frac{\partial h}{\partial \bx}(0)=0$.
Substituting (\ref{princ_eig}) in  equation (\ref{eig_koopmang})  and comparing terms, we obtain
\begin{align}\bw_\lambda^\top \bA=\lambda \bw_\lambda^\top,\label{linear_eig}
\end{align}
i.e., $\bw_\lambda$ is the left eigenvector of $\bA$ with eigenvalue $\lambda$. Similarly, the nonlinear part, $h_\lambda(\bx)$, of the eigenfunction satisfies the following linear PDE 
\begin{align}\frac{\partial h_\lambda}{\partial \bx}\bff(\bx)-\lambda h_\lambda(\bx)+\bw_\lambda^\top \bff_n(\bx)=0.\label{pde}
\end{align}
We now present the path-integral approach for the computation of the nonlinear part of the principal eigenfunctions in the following theorem.
\begin{theorem}\big[\textup{Path-integral formula \cite{deka2023path}}\big] The solution for the first order linear PDE \eqref{pde} can be written as 
\begin{align}
h_\lambda(\bx)=e^{-\lambda t} h_\lambda(\bs_t(\bx))+\int_0^t e^{-\lambda \tau} \bw^\top_\lambda \bff_n(\bs_\tau(\bx))d\tau,\label{pde_soultion}
\end{align}
where $\bs_t(\bx)$ is the flow of the system (\ref{eq:dynamics}). 
\end{theorem}

The solution $h_\lambda(\bx)$ obtained through equation \eqref{pde_soultion} is an implicit expression, which needs the value of $h_\lambda$ at a point $\bs_t(\bx)$. However, in certain cases, one can resolve this issue by considering the limit $t \rightarrow \infty$. This is presented as follows.

\begin{theorem}\big[\textup{Convergence rate of nonlinear term \cite{deka2023path}}\big]\label{theorem_main2}
For the dynamical system (\ref{eq:dynamics}), let the origin be an asymptotically stable equilibrium point with the domain of attraction $\cD$ and let $\bA$ be Hurwitz. Furthermore, suppose all the eigenvalues of the $\bA$ satisfy 
\begin{align}
-{\rm Re}(\lambda)+2{\rm Re}(\lambda_{max})<0,\label{eigendistribution}
\end{align}
where $\lambda_{max}$ is the eigenvalue closest to the imaginary axis and in the left half plane. Let $h_\lambda$ be the solution of PDE (\ref{pde}) as given in (\ref{pde_soultion}). 
Then, 
\begin{align}
\lim_{t\to \infty}e^{-\lambda t}h_\lambda(\bs_t(\bx))=0,\;\;\;\forall \bx\in \cD\label{condition_convergence}
\end{align}
if 
$h_\lambda(\bx)$ is purely nonlinear function of $\bx$ i.e., $\frac{\partial h_\lambda}{\partial \bx}(0)=0$. 
\end{theorem}

This theorem provides sufficient conditions on the distribution of the principal eigenvalues of a stable system, such that the nonlinear part of a principal eigenfunction $\phi_\lambda$ decays at a rate faster than $\lambda$, leading to the first term on the right-hand side of \eqref{pde_soultion} to converge to zero asymptotically. The result is trivially extendable to anti-stable systems by reversing the flow of the system. This leads to an easily computable expression for the principal eigenfunction, as summarized in the following theorem.

\begin{theorem}\big[\textup{Eigenfunctions of stable systems\cite{deka2023path}}\big]
\label{theorem_stable}
Consider the dynamical system (\ref{eq:dynamics}) with origin asymptotically stable and with the domain of attraction $\cD$. Let the eigenvalue $\lambda$ of matrix $\bA$ satisfy condition (\ref{eigendistribution}). Then the principal eigenfunction, $\phi_\lambda$, corresponding to eigenvalue $\lambda$, is well defined in the domain $\cD$ and is given by following path-integral formula:
\begin{align}
\phi_\lambda(\bx)=\bw_\lambda^\top \bx+\int_0^\infty e^{-\lambda t}\bw_\lambda^\top \bff_n(\bs_t(\bx))dt\label{pathintegralformula}
\end{align}
where $\bw_\lambda$ \textcolor{black}{satisfies} $\bw_\lambda^\top \bA=\lambda \bw_\lambda^\top$. 
\end{theorem}

\textcolor{black}{
\begin{remark}
    The path-integral formula \eqref{pathintegralformula} bears some resemblance to the Laplace average-based computation technique proposed in \cite{mauroy2013isostables}\cite{mohr2014construction}. Indeed, computation of Laplace averages also involves evaluating an observable along system trajectories, exponentially weighted over an infinite time horizon. However, finding observables with converging Laplace averages is non-trivial, and as noted in literature, these computations can become numerically problematic. By contrast, the integral in equation \eqref{pathintegralformula} is evaluated along an observable which is known a priori, allowing us to obtain rigorous convergence guarantees \cite{deka2023path}. Interestingly, similar conditions as in \cite{deka2023path} have been shown to be sufficient for the convergence of Laplace averages of certain observables \cite{deka2023supervised}.
\end{remark}}

 
 Computation of the Koopman eigenfunction using theorem \ref{theorem_stable} is restricted to systems with stable and anti-stable equilibrium points. For systems with a saddle-type equilibrium point, the path-integral formula in (\ref{pathintegralformula}) will not work, as we shall discuss in section \ref{sec:path_int_saddle}. Given the usefulness of the Koopman eigenfunction for saddle point systems, as described in section \ref{sec:Intro}, it is of interest to develop an approach for its approximation. Toward this goal, we propose an approach based on the finite-time evaluation of the path-integral formula for the computation of the eigenfunctions in the next section. 

\section{Path integrals for saddle point systems}\label{sec:path_int_saddle}
We note that, although the path-integral \eqref{pde_soultion} provides a solution to nonlinear part $h_\lambda(\bx)$ of the principal eigenfunction $\phi_\lambda(\bx)$, the value of $h_\lambda$ at any point $\bs_t(\bx)$ is usually unknown. While theorem \ref{theorem_main2} resolves this issue for stable (and anti-stable systems) by taking the limit $t\rightarrow \infty$, but for saddle point systems $e^{-\lambda t}h_\lambda(\bs_t(\bx))$ can easily diverge since both the exponential term (with ${\rm Re}(\lambda)<0$) and $\bs_t(\bx)$ diverge as $t\rightarrow \infty$. Thus, we require a known terminal boundary condition on $h_\lambda$ in order to utilize path-integral formula \eqref{pde_soultion} over a finite time horizon.


We now present two ways to estimate the value of $h_\lambda$ over a boundary set $\mathcal{S}$. The first method is a data-driven approach that relies on trajectory data sampled locally near the set $\mathcal{S}$. The other approach is model-based and involves transforming the analytical vector fields of the system into a new system such that the nonlinear part of the Koopman eigenfunctions, $h_\lambda$, for this new system are approximately zero on the set $\mathcal{S}$.

\subsection*{\textbf{Method A: Local-EDMD for approximating terminal boundary condition}}\label{sub:method1}

EDMD is a well-known algorithm for estimating Koopman eigenfunctions and/or finite-dimensional approximation of the Koopman operator from trajectory data \cite{williams2015data}. Typically for systems with an attractor set $\mathcal{A}$, EDMD is employed on trajectory data sampled from the domain of attraction $\cD$ of set $\mathcal{A}.$ The resulting approximation of the Koopman operator and its spectrum is expected to hold on this entire set of $\cD.$

Consider a set of basis functions $\theta_i(\bx) \in \mathcal{F}, \; i=1,2,..,N$ with $\theta(\bx)\doteq\left[\theta_1(\bx),\theta_2(\bx),\ldots,\theta_N(\bx)\right]^\top$.
If we have $M$ pairs of trajectory snapshots at some uniform sampling time $\tau$, in form of  $(\bx_i,\by_i)\in \cD \times \cD$ with $\by_i = \bs_\tau(\bx_i)$ for $i=1,2,...,M$, we recall that the EDMD algorithm yields a finite-dimensional approximation of the Koopman Operator $\mathcal{K}_\tau$ restricted to the subspace spanned by $\theta_i(\bx) \in \mathcal{F}, \; i=1,2,..,N$. This is expressed in terms of the Koopman matrix $K$, by solving the following:
\begin{equation}\label{eq:EDMD}
K = \underset{A \in \mathbb{C}^{N \times N}}{\argmin} \|\theta(Y) - A\theta(X)\|^2_F 
\end{equation}
where $\|\cdot\|_F$ denotes the Frobenius norm, and $\theta(\bx) = [\theta(\bx_1), \theta(\bx_2), \cdots, \theta(\bx_M)]$ and $\theta(Y) = [\theta(\by_1),\theta(\by_2),\cdots,\theta(\by_M)]$ are $N \times M$ matrices. The closed-form expression for $K$ that minimizes the least-squares problem \eqref{eq:EDMD} is $K = \theta_{XY}\theta_{XX}^{\dagger}$, where $\dagger$ denotes the pseudo-inverse, $\theta_{XY} = \theta(Y)\theta(X)^\top$ and $\theta_{XX} = \theta(X)\theta(X)^\top$. Koopman eigenfunctions can then be extracted from the $K$ matrix using its left eigenvectors. That is, if $\nu \in \mathbb{R}^N$ is a left eigenvector of $K$, corresponding to some eigenvalue $\mu$, then clearly $\nu^\top \theta(\bx)$ approximates the Koopman eigenfunction, since $\nu^\top \theta(\bs_\tau(\bx)) \approx \nu^\top K \theta(\bx) = \mu \nu^\top \theta(\bx)$ for all $\bx \in \cD.$

Note that one can restrict the trajectory samples $(\bx_i,\by_i), \, i=1,\ldots,N,$ to some small neighborhood of a set $\mathcal{S},$ leading to local approximations of Koopman eigenfunctions confined to the small region from where the EDMD samples are drawn. Given a fixed set of basis functions $\theta_i,\,i=1,\ldots,N$, one can argue that EDMD estimates of the eigenfunctions are more accurate when confined to a smaller validity domain. For example, given a two dimensional dynamical system that has a stable origin and a stable Koopman eigenfunction 
\begin{equation}\label{eq:edmd_ex}
    \psi_\lambda(\bx) = \bx_1^4 + 2\bx_1^2 + \bx_1\bx_2 - \bx_2 + \bx_2^4,
\end{equation}
let us say we have EDMD basis functions comprising of monomials of degree up to 2. Clearly, with these basis functions, the estimate of the eigenfunction \eqref{eq:edmd_ex} will be inaccurate away from the origin. However, when restricted to the manifold $\mathcal{S}= \left\{\bx \;\lvert\; \bx_1^4 + \bx_2^4 = 1 \right\}$, we can exactly estimate the eigenfunction \eqref{eq:edmd_ex} with our given monomial basis functions since
\[
    \psi_\lambda(\bx) = 2\bx_1^2 + \bx_1\bx_2 - \bx_2 + 1\; \forall \bx \in \mathcal{S},
\]
with a nonlinear part $h_\lambda(\bx)=2\bx_1^2 + \bx_1\bx_2 + 1$ on the manifold $\mathcal{S}$. Once we have the estimate of $h_\lambda(\bx)$ on set $\mathcal{S}$, we can use the path-integral \eqref{pde_soultion} to find $h_\lambda(\bx)$ on an larger, extended domain (more precisely, on the set $\left\{\bs_t(\bx) \; \vert \; \bx \in \mathcal{S}, \, t\in \mathbb{R}\right\}$). Note that we require non-recurrence assumption on the set $\mathcal{S}$, as explained in \cite{korda2020optimal}.

Thus, the following theorem provides an approach for computing the approximation of the Koopman eigenfunction over some set $\bar \cX$ based on prior known estimates of the eigenfunction in some other known region $\cS$. 
\begin{theorem} Let $\cS\subset \cX$ and $\bar \cX\subseteq \cX$ be such that for all $\bx\in \bar\cX$ there exists a $\bar t(\bx)\in \mR$ for which $\bs_{\bar t(\bx)}(\bx)\in \cS$. Furthermore, \textcolor{black}{consider a function $\hat h_\lambda(\bx)$ such that}
\[\|h_\lambda(\bx)-\hat h_\lambda(\bx)\|\leq \epsilon_\cS\]
for all $\bx\in \cS$. Then for any point $\bx\in \bar \cX$, we have

\begin{align}\|\tilde h_\lambda(\bx)-h_\lambda(\bx)\|\leq \epsilon_\cS e^{-{\rm Re}(\lambda) \bar t(\bx)},\label{results}
\end{align}
where 
\small
\begin{align}
\tilde h_\lambda(\bx)=e^{-\lambda \bar t(\bx)}\hat h_\lambda(\bs_{\bar t (\bx)}(\bx))+\int_0^{\bar t(\bx)} e^{-\lambda \tau}\bw_\lambda^\top \bff_n(\bs_\tau(\bx))d\tau\label{modformula}
\end{align}
\end{theorem}

\begin{proof}
We know that for any point $\bx\in \bar \cX$, there exists a time $\bar t(\bx)$ for which $\bs_{\bar t(\bx)}(\bx)\in \cS$, hence using the PDE solution formula (\ref{pde_soultion}), we have 
\small\begin{align}
h_\lambda(\bx)=e^{-\lambda \bar t(\bx)} h_\lambda(\bs_{\bar t(\bx)}(\bx))+\int_0^{\bar t(\bx)} e^{-\lambda \tau} \bw^\top_\lambda \bff_n(\bs_\tau(\bx))d\tau,\label{estimateS}
\end{align}
The result (\ref{results}) then follows from (\ref{modformula}) and (\ref{estimateS}).
\end{proof}

\begin{remark}
    For unstable eigenfunctions where ${\rm Re}(\lambda)>0$, the approximation error bound \eqref{results} can be seen to exponentially diminish for points farther away from the set $\mathcal{S}$. Moreover, once $h_\lambda$ is estimated on $\mathcal{S}$ using a given choice of basis functions, the path-integral formula \eqref{pde_soultion} allows us to estimate $h_\lambda$ on $\bar \cX$ in a basis-free manner. A similar approach of combining local EDMD with Laplace averages evaluated along global trajectories was recently proposed in \cite{deka2023supervised}, albeit for systems with a stable equilibrium.
\end{remark}

\subsection*{\textbf{Method B: Enforcing zero terminal boundary condition using vector field transformation}}\label{sub:method2}
Let us consider the dynamical system \eqref{eq:dynamics} again. For our presentation in this subsection, we are mainly interested in the computation of unstable principal eigenfunctions for such systems (stable principal eigenfunctions can be similarly found by reversing the flow of the dynamics). 

One can smoothly transform the vector field $\bff(\bx)$ as follows
\begin{equation*}
    \Tilde{\bff}(\bx) \doteq \bff(\bx) + \sigma(\|\bx\| - r)\big(\bA\bx - \bff(\bx)\big),
\end{equation*}
where the scalar function $\sigma: \mathbb{R} \rightarrow (0,1)$ is a smooth approximation of the unit step function, leading to $\Tilde{\bff}(\bx) \approx \bff(\bx)$ for $\|\bx\|<<r$ and $\Tilde{\bff}(\bx) \approx \bA\bx$ for $\|\bx\|>>r$. For the sake of concreteness, in the rest of the paper, we shall take this scalar function $\sigma(z)$ to be $\frac{1}{2}(1+\tanh(az))$ where $a>0$ is some large constant. We shall further write $\sigma(\|\bx\|-r)$ simply as $\sigma_r$ for brevity. That is, we take
\begin{equation}\label{eq:mod}
    \Tilde{\bff}(\bx) = \bff(\bx) + \frac{1}{2}\big(1+\tanh(a(\|\bx\|-r))\big)\big(\bA\bx - \bff(\bx)\big).
\end{equation}
\textcolor{black}{\vspace{-0.4cm}
\begin{assumption}\label{assum:equilibruim}
    We assume that the construction \eqref{eq:mod} does not introduce any spurious equilibrium points or $\omega$-limit sets. That is, $\bx=0$ is the only equilibrium of $\Tilde{\bff}(\bx)$.
\end{assumption}}

\begin{remark}
    Given that the original system has a unique equilibrium at $\bx=0$, a sufficient condition for uniqueness of equilibrium at zero for the system $\dot{\bx}=\Tilde{\bff}(\bx)$ is that $\bff^\top \bA\bx$ must be non-negative, since $\bff^\top \bA\bx > 0 \implies \|\tilde{\bff}\|^2 > 0$ for $\bx \neq 0$. For ensuring that $\tilde{\bff}$ does not have other types of $\omega-$ limit sets, however, a more thorough analysis is needed.
\end{remark}

Note that, for the transformed system $\dot{\bx} = \Tilde{\bff}(\bx)$, the origin remains as an hyperbolic equilibrium. Furthermore, both $\Tilde{\bff}$ and $\bff$ have the same principal Koopman eigenvalues, since $\nabla_\bx\Tilde{\bff}(0) = \nabla_\bx \bff(0) = \bA$. Thus, we can state the following proposition.
\begin{proposition}\label{prop:main1}
Given a system \eqref{eq:dynamics} and the transformation \eqref{eq:mod}, any $\mathcal{C}^1$ principal eigenfunction $\phi_\lambda$ of system \eqref{eq:dynamics} is also approximately a principal eigenfunction for the system $\dot{\bx}=\tilde{\bff}(\bx)$ in the region $\|\bx\|<r$. That is, for any positive constants $\epsilon_1$ and $\epsilon_2<r$, we can choose the constant $a>0$ such that
\begin{equation}\label{eq:approx_bound}
    \| \tilde{\bff}^\top \nabla_\bx \phi_\lambda - \lambda \phi_\lambda \| \le \epsilon_1
\end{equation}
whenever $\|\bx\|\le r - \epsilon_2.$
\end{proposition}
\begin{proof}
Since $\bff^\top \nabla_\bx \phi_\lambda = \lambda \phi_\lambda$ and $\|\bx\|\le r - \epsilon_2$, we have
\begin{align*}
     &\| \tilde{\bff}^\top \nabla_\bx \phi_\lambda - \lambda \phi_\lambda \| = \|(\bff - \Tilde{\bff})^\top \nabla_\bx \phi_\lambda\|\\
     &\le \frac{1}{2}\big(1+\tanh(a(\|\bx\|-r))\big) \|\bff_n(\bx)^\top\nabla_\bx \phi_\lambda\|\\
     &\le \frac{1}{2}\big(1 - \tanh(a\epsilon_2)\big) \beta,
\end{align*}
where $\beta = \sup_{\|\bx\|\le r-\epsilon_2}\|\bff_n(\bx)^\top\nabla_\bx \phi_\lambda\|$ is finite due to continuity of $\bff_n(\bx)^\top\nabla_\bx \phi_\lambda$ and compactness of the region $\|\bx\|\le r - \epsilon_2$. Thus, any $a > \frac{1}{\epsilon_2}\tanh^{-1}\big(1 - \frac{2\epsilon_1}{\beta}\big)$ leads to the inequality \eqref{eq:approx_bound}. This completes our proof.
\end{proof}
While proposition \ref{prop:main1} establishes that the eigenfunctions of the original and the transformed dynamical systems within some domain inside $\|\bx\|=r$ can be made arbitrarily close, the next proposition provides a closed-form approximation of the principal eigenfunctions of the transformed system.
 \begin{proposition}\textcolor{black}{\big[\textup{Eigenfunction boundary condition}\big]}\label{prop:main2}
     Consider again a system \eqref{eq:dynamics} and the transformation \eqref{eq:mod}. Given a fixed $\epsilon_1,\, \epsilon_1>0$ and eigenvalue $\lambda$ of matrix $\bA$, we can choose the constant $a>0$ such that
     \[
        \| \Tilde{\bff}^\top \nabla_\bx (\bw_\lambda^\top \bx) - \lambda \bw^\top_\lambda \bx\| \le \epsilon_1,
     \]
     whenever $\|\bx\|= r + \epsilon_2.$ In other words, $\tilde{\phi}_\lambda(\bx) = \bw_\lambda^\top \bx$ approximates the Koopman eigenfunction for the system $\dot{\bx}=\tilde{\bff}$ on the sphere $\|\bx\| = r + \epsilon_2$, if $a$ is sufficiently large.
 \end{proposition} 
 
\begin{proof}
We have
\begin{align}\label{eq:approx2}
    \begin{split}
        &\| \Tilde{\bff}^\top \nabla_\bx (\bw_\lambda^\top \bx) - \lambda \bw^\top_\lambda \bx\| \\
        &\le  (1-\sigma_r)\| \bff^\top \bw_\lambda\| + \|\sigma_r \bx^\top \bA^\top \bw_\lambda - \lambda \bw^\top_\lambda \bx \| \\
        &= (1-\sigma_r)\| \bff^\top \bw_\lambda\| + \|(\sigma_r - 1)\lambda \bw^\top_\lambda \bx \| \\
        &= (1-\sigma_r)\big( \| \bff^\top \bw_\lambda\| + \|\lambda \bw^\top_\lambda \bx\|\big)
    \end{split}
\end{align}
where we recall that $\sigma_r$ denotes $\sigma(\|\bx\|-r)$, which we take as $\frac{1}{2}\big(1+\tanh(a(\|\bx\|-r))\big).$ Thus, on the sphere $\|\bx\|=r + \epsilon_2$, inequality \eqref{eq:approx2} gives us
\begin{gather*}
    \| \Tilde{\bff}^\top \nabla_\bx (\bw_\lambda^\top \bx) - \lambda \bw^\top_\lambda \bx\| \le \\
    \frac{1}{2}\big(1-\tanh(a\epsilon_2)\big) \|\bw_\lambda\|\bigg( \|\lambda\|(r+\epsilon_2) + \sup_{\|\bx\|=r+\epsilon_2} \| \bff\|\bigg).
\end{gather*}
Thus, similar to proposition \ref{prop:main1}, we can make an appropriate choice of $a$ such that the right-hand side of the above inequality is less than $\epsilon_1.$ Our proof is now complete.
\end{proof}


While direct application of path-integral formula to saddle point systems may not always be possible, the key idea we propose here is to modify the original system just enough using \eqref{eq:mod}, so that its eigenfunctions approximate the original system's eigenfunctions in a certain region of interest within the state-space (proposition \ref{prop:main1}), while rendering path-integrals computationally amenable for the transformed system. Notice that, with the construction of vector fields $\Tilde{\bff}$, one can employ path-integrals to approximate the principal Koopman eigenfunctions in a domain inside $\|\bx\| = r$, by using the fact that the eigenfunction $\tilde{\phi}_\lambda(\bx) \approx \bw^\top_\lambda \bx$ on some sphere $\|\bx\|=R$ of radius $R>r$ (proposition \ref{prop:main2}). In other words, its corresponding nonlinear part satisfies the boundary condition $\tilde{h}_\lambda(\bx) \approx 0$ on $\|\bx\| = R.$ Thus, given any point inside $\|\bx\|<r$ if there exists a time $T(\bx) >0$ such that $\|s_T(\bx)\|=R$, we get
\[
h_\lambda(\bx) = \underbrace{\cancel{e^{-\lambda T}\tilde{h}_\lambda(s_T(\bx))}}_{\approx 0} + \int_0^{T(\bx)}  e^{-\lambda \tau}w^\top_\lambda \Tilde{\bff}_n(s_\tau(\bx))d\tau.
\]

\noindent Here, the flow-map $s_t(\cdot)$ corresponds to the transformed vector field $\tilde{\bff}$. \textcolor{black}{The assumption \ref{assum:equilibruim} is important here, since we require that} the trajectories reach $\|\bx\| = R$ from almost everywhere (a.e.) in some subset of $\|\bx\|<r.$ Otherwise, the time $T(x)$ may become ill-defined. 

\section{Numerical experiments}

This section demonstrates our approach for extending the path-integral formula towards computation of eigenfunction for saddle point systems using finite-time integrals. The first example illustrates how this work overcomes the limiting assumptions in \cite{deka2023path} for path-integral based eigenfunction computations using method A in section \ref{sec:path_int_saddle}. The second example focuses on a Hamiltonian system (which is a particular type of saddle point system) to highlight the practical utility of Koopman eigenfunction computations within the context of nonlinear optimal control, while showcasing method B. 

\subsection*{\textbf{Example A:} Path-integrals for saddle point systems with EDMD-estimated boundary conditions}
In this example, we consider a 2-dimensional dynamical system with states $\bx = [\bx_1, \bx_2]^\top$ as follows:
\begin{equation}\label{eq:analytical_vf}
\dot{\bx} = 
 \left[ \begin{array}{c}
 \frac{\left(7.5 \bx_2^2+5.0\right)\left(\bx_1^3+\bx_1+\sin \left(\bx_2\right)\right)+\left(-\bx_1+\bx_2^3+2 \bx_2\right) \cos \left(\bx_2\right)}{9 \bx_1^2 \bx_2^2+6 \bx_1^2+3 \bx_2^2+\cos \left(\bx_2\right)+2} \\
\\
\frac{2.5 \bx_1^3+2.5 \bx_1-\left(3 \bx_1^2+1\right)\left(-\bx_1+\bx_2^3+2 \bx_2\right)+2.5 \sin \left(\bx_2\right)}{9 \bx_1^2 \bx_2^2+6 \bx_1^2+3 \bx_2^2+\cos \left(\bx_2\right)+2}
\end{array}\right].
\end{equation}

The vector field for this system is smooth everywhere on $\mathbb{R}^2$, with a saddle point at the origin as shown in figure \ref{fig:Result1}(a). \textcolor{black}{One can verify that} the two principal eigenpairs of this system are given by
\begin{gather*}
\phi_{\lambda_1}= \bx_1 - 2\bx_2 - \bx_2^3,\quad  \lambda_1 = -1 \text{ and}\\
\phi_{\lambda_2} = \bx_1 + \sin(\bx_2) + \bx_1^3,\quad \lambda_2 = 2.5.
\end{gather*}
We note that the principal eigenpairs of this system do not follow the assumption regarding the spectral distribution or the assumption on the boundedness of $h_\lambda$ in \cite{deka2023path}, and thus the infinite horizon path-integral formula proposed in the prior work \cite{deka2023path} is not suitable for eigenfunction computations.

\begin{figure}[!htp]
    \centering
  \begin{minipage}[b]{\linewidth}
    \centering \hspace{-0.4cm}
    \includegraphics[width=0.5\linewidth]{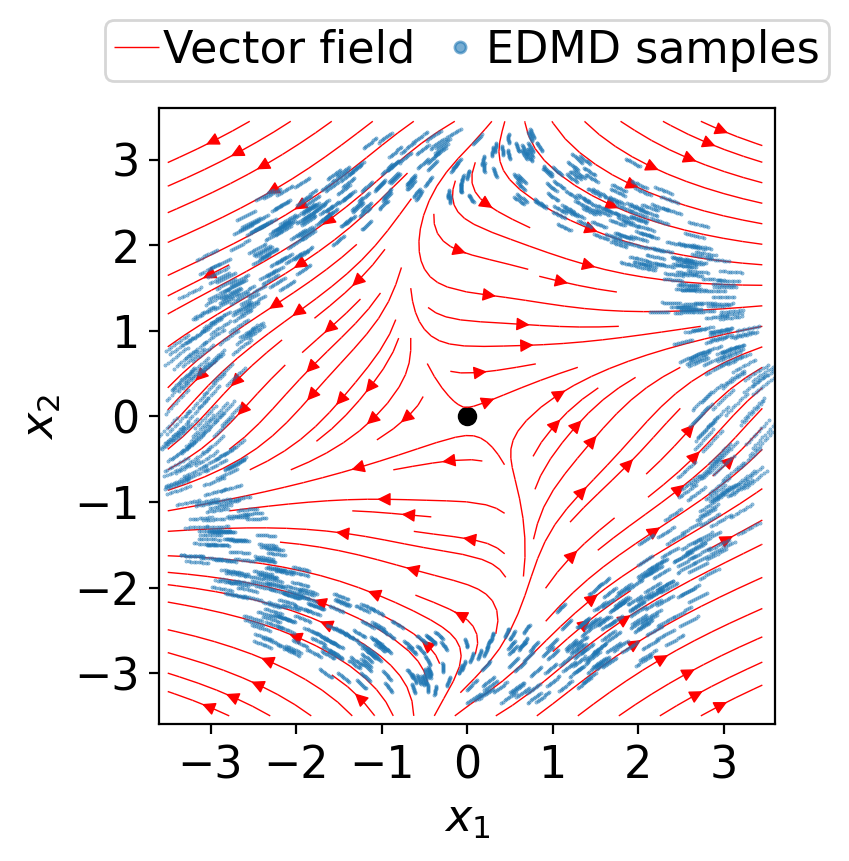} 
    \caption*{(a) Saddle point system, equation \eqref{eq:analytical_vf}. } 
  \end{minipage}\vspace{0.4cm}
 \begin{minipage}[b]{\linewidth}
    \hspace{-0.2cm}
    \includegraphics[width= 0.5\columnwidth]{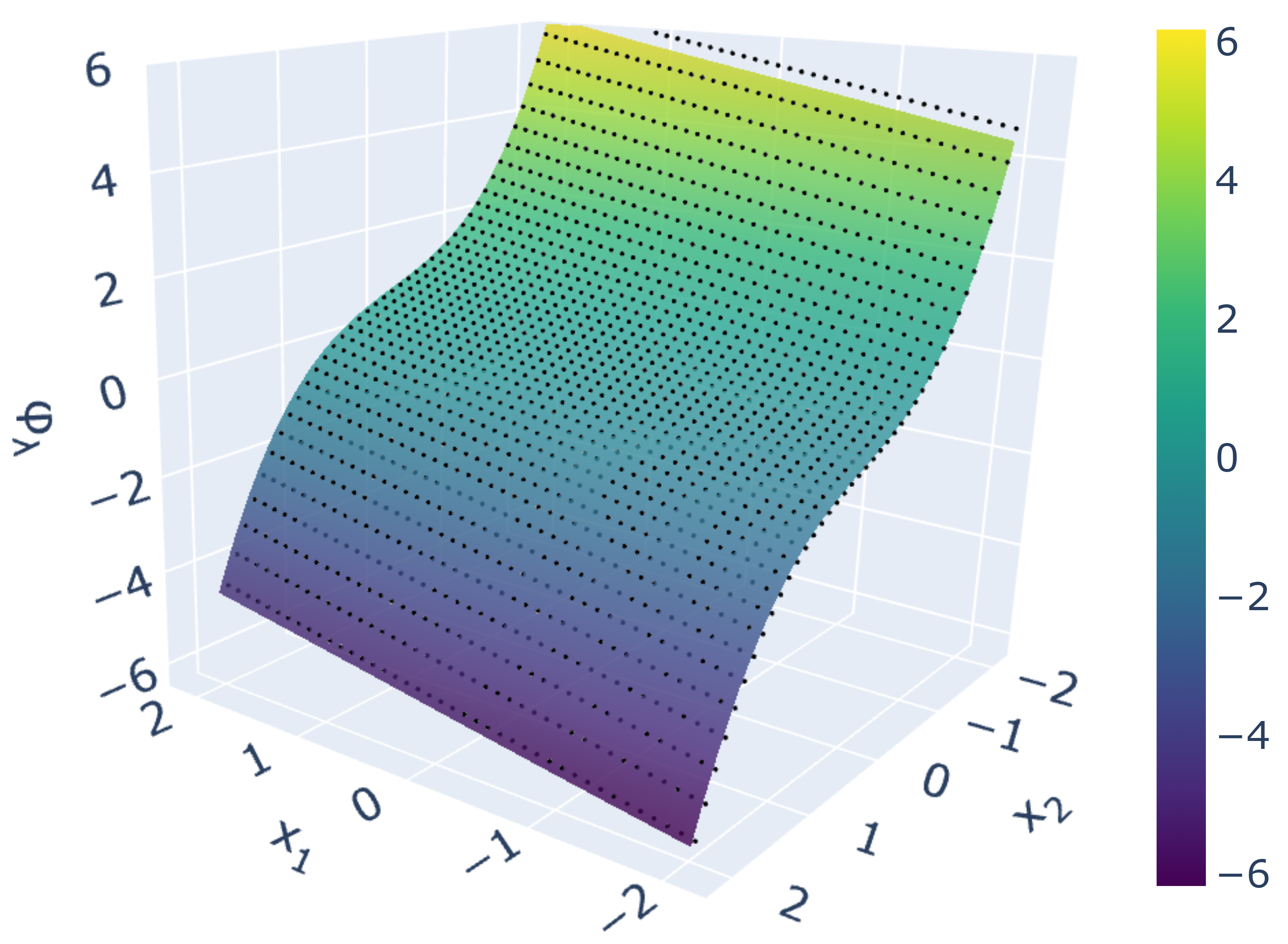}
    \hspace{-0.2cm}
    \includegraphics[width= 0.5\columnwidth]{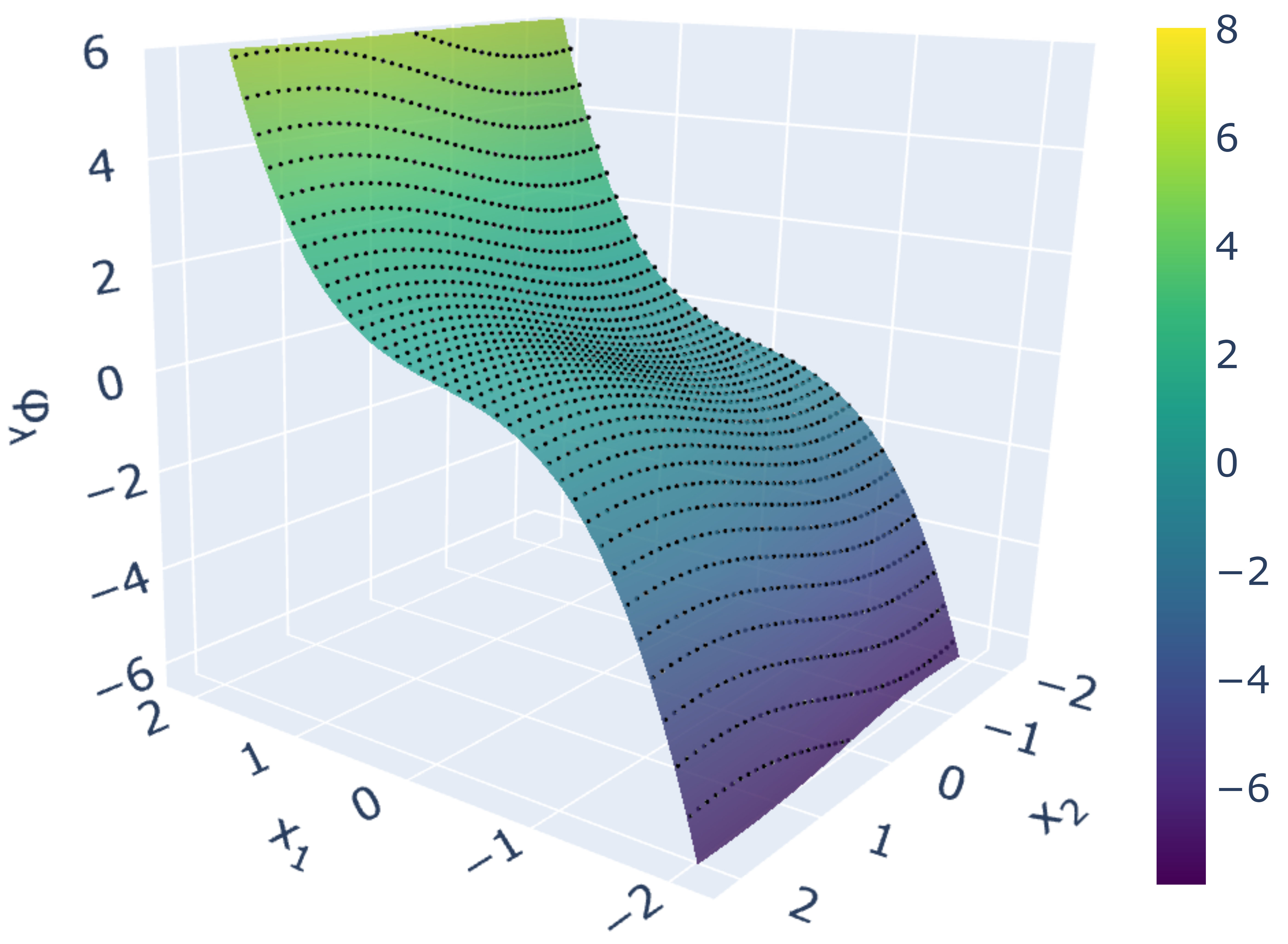}
    \caption*{(b) Computed (surface) vs actual (dots)
            eigenfunctions.} 
    \end{minipage}
\caption{Local-EDMD enabled path-integral computation for Koopman eigenfunctions. (a) Blue dots represent trajectory samples collected for EDMD, giving us an approximation of eigenfunctions in the annular disk. This estimate is then used for the terminal value $h_\lambda(s_t(\bx))$ appearing in the path-integral formula. (b) Left pane shows stable principal eigenfunction $\phi_1$ with $\lambda_1 = -1$ (Max relative error: 0.055). Right pane shows unstable principal eigenfunction $\phi_2$ with $\lambda_2 = 2.5$ (Max relative error: 0.039).}
\label{fig:Result1}
\end{figure}

In order to utilize the finite-time path-integrals in the computation of principal eigenfunctions on a compact set $C=[-2,2]^2$ from trajectory data, we first use EDMD to compute the eigenfunction values on the \textit{boundary} of a sphere of radius $3$ centered at the origin (i.e., $B=\left\{ \bx \, \vert \, \|\bx\|=3 \right\}$). For this preliminary EDMD step, we randomly sample $800$ trajectories of length $10$ timesteps, only from a small neighbourhood of $B$ given by $\mathcal{N}_B=\left\{\bx \, \vert \, \inf_{\bx' \in B} \|\bx-\bx'\| \le 0.4 \right\}$, as shown in figure \ref{fig:Result1}(a). We choose a monomial basis of maximum degree 10. We note that the EDMD estimated eigenfunctions from such a sampling are expected to be accurate only locally within $\mathcal{N}_B$, and not on the whole of $C=[-2,2]^2$. After extracting the nonlinear part of the local EDMD estimate, we apply \eqref{pde_soultion} to compute the eigenfunctions on the entire set $C$. For a given point $\bx\in C$, the integration time is taken to be $\min \left\{ t \, \vert \, \|\bs_t(\bx)\|=3\right\}$, that is, the integration is terminated when the trajectories reach the set $B.$ Figure \ref{fig:Result1}(b) demonstrate this for $\phi_{\lambda_1}$ and $\phi_{\lambda_2}$ respectively. 
Note that in order to estimate these entirely using EDMD, one would require trajectory samples from the entire domain $C$, and from our experimental comparison, we find that this data requirement is an order of magnitude higher than our current EDMD sampling on the set $\mathcal{N}_B$, and furthermore yields a larger (i.e., poorer) relative error between the estimated eigenfunction and the actual one on the domain $C$.

\subsection*{\textbf{Example B:} Path-integrals for Hamiltonian systems using transformed vector fields}
Let us consider an optimal control problem \cite{guo2022tutorial} as follows:
\begin{align}\label{eq:ocp}
\begin{split}
    &\text{Objective: }\min_{\bx_1,u} \int_0^\infty \left(\bx_1^2 + u^2\right)dt\\
    &\text{Dynamics: } \dot{\bx}_1 = -\bx_1^3 + u,
    \end{split}
\end{align}
\noindent where $\bx_1\in \mR$ is a scalar state and $u \in \mR$ is the control input. The solution to the optimal control problem can be obtained by solving the Hamilton Jacobi equation. In \cite{vaidya2022spectral}, a Koopman spectral approach is developed for solving optimal control problems. The solution to the Hamilton Jacobi equation is obtained from the Lagrangian submanifold. The main results of \cite{vaidya2022spectral} show that the Lagrangian submanifold can be obtained from the joint zero-level curve of the eigenfunction of the Koopman operator associated with the Hamiltonian system. The corresponding Hamiltonian system for \eqref{eq:ocp} is constructed using the Hamiltonian function defined by $H(\bx_1,\bx_2,u) \doteq \bx_1^2 + u^2 -\bx_1^3\bx_2 + \bx_2u$, \textcolor{black}{where $\bx_2 \in \mR$ is the co-state variable}. The Hamiltonian is minimized w.r.t. $u$ to obtain $u^\star=-\frac{1}{2}\bx_2$. Substituting for optimal $u^\star$ in the Hamiltonian, we obtain the pre-Hamiltonian $\bar H(\bx_1,\bx_2, u^\star(\bx))$. The Hamiltonian system associated with $\bar H$ is then given by:
\begin{equation}\label{eq:Ham_example}
    \left[\begin{array}{cc}
         \dot{\bx}_1(t)  \\
         \dot{\bx}_2(t) 
    \end{array}\right] = 
    \left[\begin{array}{cc}
         - \bx_1^3(t) - \frac{1}{2}\bx_2(t)\\
         -2\bx_1(t) + 3 \bx_1^2(t)\bx_2(t)
    \end{array}\right] = \bff(\bx),
\end{equation}
 Clearly, $\bx=[\bx_1,\bx_2]^\top = 0$ is a saddle point equilibrium with eigenvalues of $\bA=\nabla_\bx \bff(0)$ being $\lambda  = \pm 1$. We choose the function $\sigma(\cdot)$ as $\frac{1}{2}\left[1 + \tanh(10(\cdot))\right]$, with $r=4$ while constructing the modified vector field $\Tilde{\bff}$ using equation \eqref{eq:mod}. Path-integral computations performed directly on system \eqref{eq:Ham_example} do not converge to a finite value (except on the stable manifold). However, for the dynamics $\Tilde{\bff}$, this issue isn't observed and the path-integrals converge before the trajectories escape to infinity.

Figure \ref{fig:Ham_eig}(a) shows the original vector field \eqref{eq:Ham_example}, and how we transform it such that the nonlinear part of the eigenfunction corresponding to the transformed vector field is approximately zero beyond $r=4.$ The eigenfunctions can then be computed using path-integrals with such a transformation, as shown in \ref{fig:Ham_eig}(b).

\begin{figure}[ht] 
  \begin{minipage}[b]{0.5\linewidth}
    \centering
    \includegraphics[width=\linewidth]{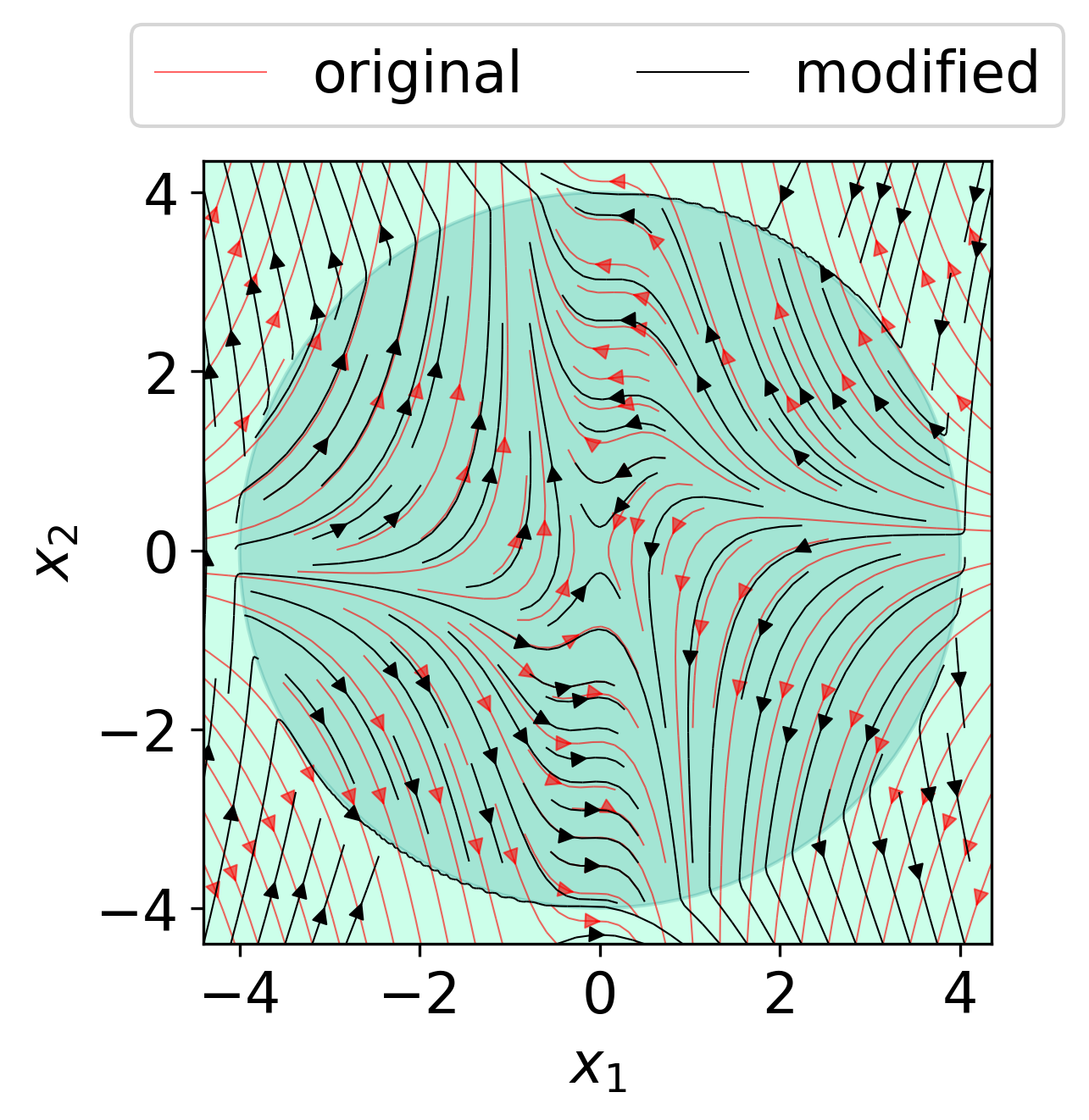} 
    \caption*{\hspace{1em}(a)} 
    \vspace{4ex}
  \end{minipage}
  \begin{minipage}[b]{0.5\linewidth}
    \centering \hspace{-0.4cm}\vspace{0.5cm}
    \includegraphics[width=1.05\linewidth]{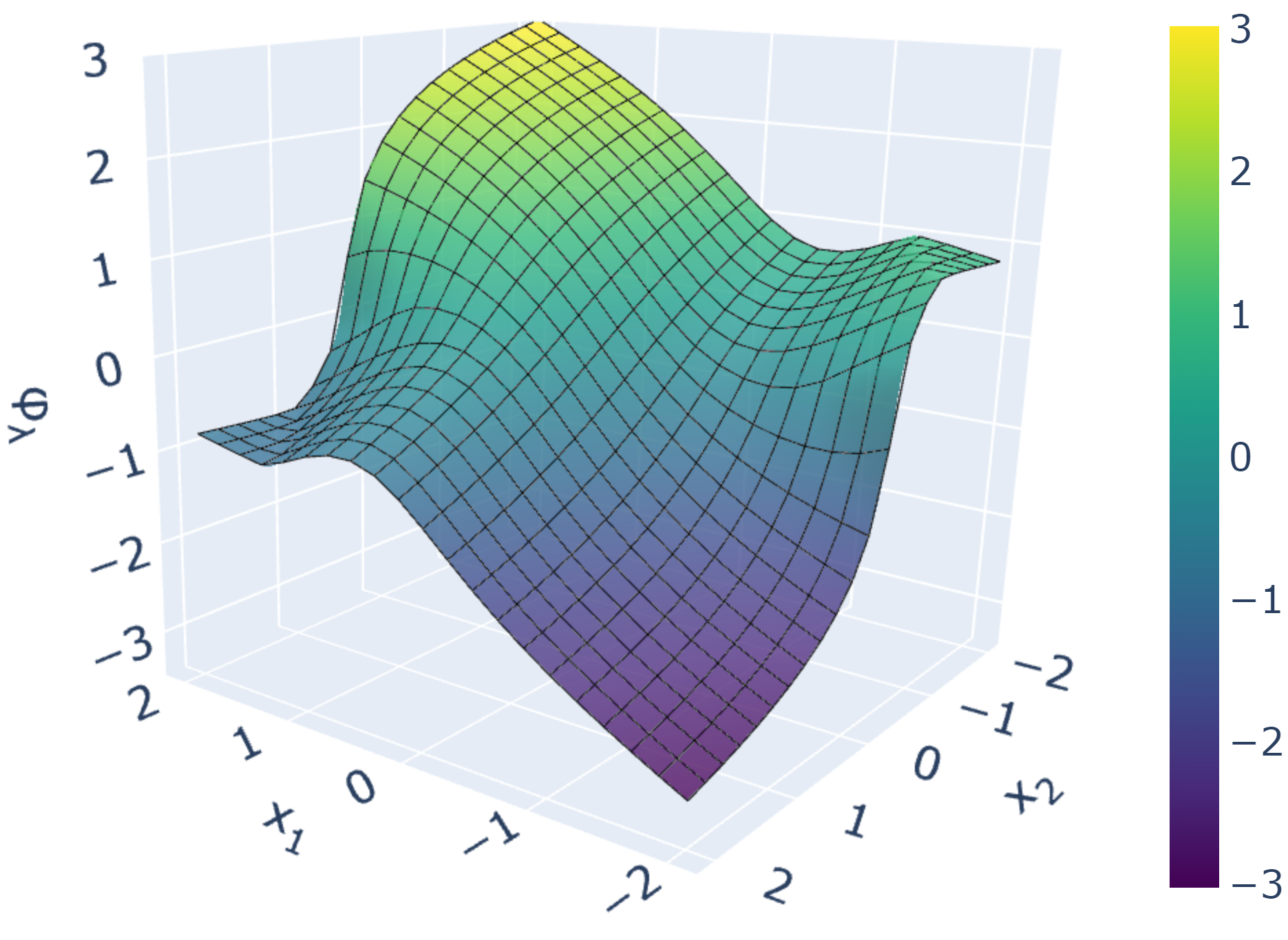} 
    \caption*{\hspace{0.5em}(b)} 
    \vspace{4ex}
  \end{minipage} 
  \vspace{-1cm}
  \caption{(a) Setting path-integral boundary conditions. The darker disk of radius $r=4$ is where the two vector fields are approximately alike, and consequently their eigenfunctions restricted to $\|\bx\|<r$ are also alike. (b) Computed unstable Koopman eigenfunction ($\lambda=1$). Path-integrals along trajectories starting inside $[-2,2]\times[-2,2]$ converge after a finite number of integration steps. Alternatively, one can choose to integrate until a given trajectory leaves the $r=4$ disk.}
    \label{fig:Ham_eig}
\end{figure}

The zero level-set of the unstable Koopman eigenfunctions of \eqref{eq:Ham_example} can be expressed as a manifold of the form $\mathcal{M}_0 \doteq \left\{ (\bx_1,\bx_2) \; \vert \; \bx_2 = \gamma(\bx_1),  \bx_1\in\mR\right\}$ \textcolor{black}{for some function $\gamma:\mathbb{R}\rightarrow \mathbb{R}$}, and the optimal trajectories of problem \eqref{eq:ocp} evolve in this manifold. The $u$ that minimizes $H(\bx_1,\bx_2,u)$ is $u^*(\bx_1,\bx_2) = -0.5\bx_2$ such that $(\bx_1,\bx_2)\in \mathcal{M}_0.$ In other words, $u^*(\bx_1,\gamma(\bx_1)) = u^*(\bx_1) = -0.5 \gamma(\bx_1)$ is the optimal controller for \eqref{eq:ocp}. Figure \ref{fig:optimal_u} shows the optimal control for problem \eqref{eq:ocp} computed from the unstable eigenfunction of system \eqref{eq:Ham_example} corresponding to $\lambda=1$. For comparison, we also plot the actual analytically derived optimal controller.

\begin{figure}
    \centering
\includegraphics[width=0.65\columnwidth]{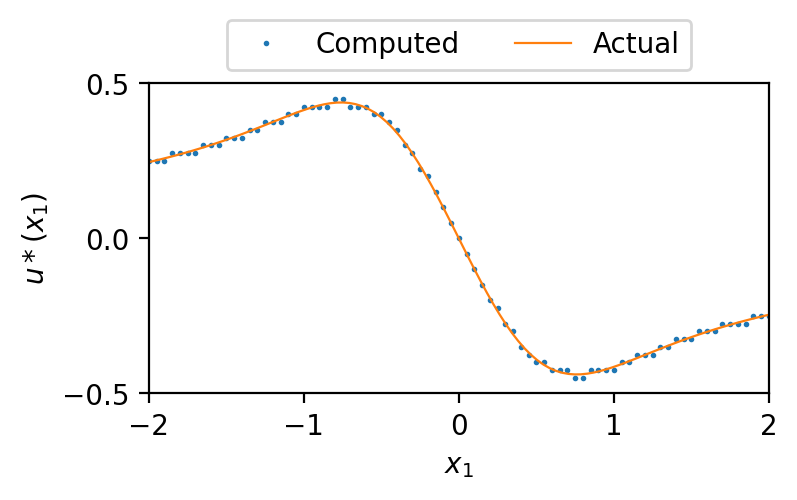}\vspace{-0.2cm}
    \caption{Optimal control extracted from the zero level-set of the unstable Koopman eigenfunction. The actual optimal control is obtained from \cite{guo2022tutorial} as $\bx_1^3 - \bx_1\sqrt{1 + \bx_1^4}.$}
    \label{fig:optimal_u}
    \vspace{-0.5cm}
\end{figure}

\section{Conclusion}
Recently developed path-integral based computation of principal Koopman eigenfunctions has shown promise in obtaining exact results since they are based on analytical expressions. However until now, their computation have been shown to converge only for specific systems, where the equilibrium is stable or anti-stable, and when the principal eigenvalues satisfy certain conditions on their spectral distribution. This paper removes these restrictions, and extends the path-integral approach to systems with saddle point equilibrium. Furthermore, in our proposed extension, the path-integrals need to be computed only over a finite horizon as opposed to the previous work, where the time horizon is infinite. The practical utility of this extension demonstrated using numerical examples, including utility in solving optimal control problems using our path-integral method.

\bibliographystyle{ieeetr}

\end{document}